\documentclass[11pt]{article}
\usepackage[margin=1in]{geometry}
\usepackage{times}
\usepackage{amsfonts}
\usepackage{amsmath,amsthm,amssymb}
\usepackage{latexsym}
\usepackage{epic}
\usepackage{epsfig}
\usepackage{hyperref}
\usepackage{verbatim}
\usepackage[justification=centering]{caption}
\usepackage{enumitem}
\usepackage{color}
\allowdisplaybreaks[1]
\usepackage[numbers,sort&compress]{natbib}
\usepackage{algorithm}
\usepackage{algorithmic}
\usepackage{setspace}

\makeatletter
\newcommand*{\rom}[1]{\expandafter\@slowromancap\romannumeral #1@}
\makeatother

\newtheorem{theorem}{Theorem}[section]
\newtheorem{corollary}[theorem]{Corollary}
\newtheorem{lemma}[theorem]{Lemma}
\newtheorem{proposition}[theorem]{Proposition}

\newtheorem{definition}[theorem]{Definition}
\newtheorem{remark}[theorem]{Remark}

\newtheorem{problem}[theorem]{Problem}

\newtheorem{fact}[theorem]{Fact}

\newtheorem*{conjecture*}{Conjecture}
\newtheoremstyle{nonindented}{1ex}{1ex}{}{}{\bfseries}{.}{.5em}{}
\newtheoremstyle{indented}{1ex}{1ex}{\itshape\addtolength{\leftskip}{0.6cm}\addtolength{\rightskip}{0.6cm}}{}{\bfseries}{.}{.5em}{}
\theoremstyle{nonindented}
\theoremstyle{indented}
\theoremstyle{plain}

\renewcommand{\hat}{\widehat}
\renewcommand{\tilde}{\widetilde}
\renewcommand{\bar}{\overline}

\DeclareMathOperator{\poly}{poly}

\def\min{\qopname\relax n{min}}
\def\max{\qopname\relax n{max}}

\newcommand{\RR}{\mathbb{R}}

\def\E{\mathcal{E}}

\def\G{\mathcal{G}}

\def\P{\mathcal{P}}
\def\Q{\mathcal{Q}}

\def\eps{\epsilon}

\newcommand{\maxi}[1]{\mbox{maximize} & {#1 } & \\}
\newcommand{\maxis}[1]{\mbox{max} & {#1 } & \\}
\newcommand{\minis}[1]{\mbox{min} & {#1 } & \\}

\newcommand{\st}{\mbox{subject to} }
\newcommand{\sts}{\mbox{s.t.} }
\newcommand{\con}[1]{&#1 & \\}
\newcommand{\qcon}[2]{&#1, & \mbox{for } #2.  \\}

\newenvironment{lp*}{\begin{equation*}  \begin{array}{lll}}{\end{array}\end{equation*}}

	\title{The Mysteries of Security Games: Equilibrium Computation Becomes Combinatorial Algorithm Design}
	
\author{
	Haifeng Xu\thanks{Supported by NSF grant CCF- 1350900 and MURI grant W911NF-11-1-0332.} \\
	Department of Computer Science\\
	University of Southern California\\
	{\tt haifengx@usc.edu}
}

\begin{document}
	\maketitle
		
	\begin{abstract}
The \emph{security game} is a basic model for resource allocation in adversarial environments. Here there are two players, a \emph{defender} and an \emph{attacker}. The defender wants to allocate her limited resources to defend critical targets and the attacker seeks his most favorable target to attack. In the past decade,  there has been a surge of research interest in analyzing and solving security games that are motivated by applications from various domains.  Remarkably, these models and their game-theoretic solutions have led to real-world deployments in use by major security agencies like the LAX airport, the US Coast Guard and  Federal Air Marshal Service, as well as non-governmental organizations. Among all these research and applications, equilibrium computation serves as a foundation.  

This paper examines security games from a theoretical perspective and provides a unified view of various security game models. In particular, each security game can be characterized by a set system $\E$ which consists of the defender's pure strategies;  The defender's best response problem can be viewed as a combinatorial optimization problem over $\E$. Our framework captures most of the basic security game models in the literature, including all the deployed systems; The set system $\E$ arising from various domains encodes standard combinatorial problems like bipartite matching, maximum coverage, min-cost flow, packing problems, etc.  Our main result shows that equilibrium computation in security games is essentially a combinatorial problem. In particular, we prove that, for any set system $\E$, the following problems can be reduced to \emph{each other} in polynomial time: (0) combinatorial optimization over $\E$; (1) computing the minimax equilibrium for zero-sum security games over $\E$; (2)  computing the strong Stackelberg equilibrium for security games over $\E$; (3) computing the best or worst (for the defender) Nash equilibrium for security games over $\E$. Therefore, the hardness [polynomial solvability] of any of these problems implies the hardness [polynomial solvability] of all the others. Here, by ``games over $\E$'' we mean the class of security games with arbitrary payoff structures, but a fixed set $\E$ of defender pure strategies. This shows that the complexity of a security game is essentially determined by the set system $\E$.
We view drawing these connections  as an important conceptual contribution of this paper.
	\end{abstract}

\section{Introduction}
The security of critical infrastructures and areas is an important concern around the world, especially given the increasing threats of terrorism. Limited security resources cannot provide full security coverage at all places all the time, leaving potential attackers the chance to explore patrolling patterns and attack the weakness. How can we make use of the limited resources to build the most effective defense against strategic attackers? The past decade has seen an explosion of research in attempt to address this fundamental question, which has led to the development of the well-known model of security games.  A \emph{security game} is a two-player game played between a \emph{defender} and an \emph{attacker}. The defender allocates (possibly randomly)  limited security resources, subject to various domain constraints, to protect a set of {\it targets}; The attacker chooses one target to attack.  This is a basic model for resource allocation in adversarial environments, and naturally captures the strategic interaction between security agencies and potential adversaries.  Indeed, these models and their game-theoretic solutions have led to  real-world deployments in use today by major security agencies.  For example, they are used by LAX airport for checkpoint placement, the US Coast Guard for port patrolling and the Federal Air Marshal Service for scheduling air marshals \cite{tambe2011};  Recently, new models and algorithms have been tested by the Transportation Security Administration  for airport passenger screening \cite{Brown16a} and by non-governmental
organizations in Malaysia for wildlife protection \cite{Fang16a}.  

Equilibrium computation is perhaps the most basic problem in security games. Indeed, there have been numerous algorithms developed for solving various security games motivated by different real-world applications (we refer the reader to \cite{tambe2011} for a review).  However, many of these algorithms are based on integer linear  programs and heuristics, which may run in exponential time or output non-optimal solutions. The computational complexity of solving these games is not well-understood. Moreover, most of the literature has focused on the computation of the strong Stackelberg equilibrium (minimax equilibrium when the game is zero-sum), which may be inappropriate when the players move simultaneously (see Section \ref{sec:SG:equ} for a more detailed discussion). In this paper, we systematically study the computational complexity of the three main equilibrium concepts adopted in security games, namely, the minimax equilibrium, strong Stackelberg equilibrium, and Nash equilibrium. However, instead of examining all the models one by one, we provide a unified view of security games that captures most of the basic models in the literature, and prove our results in this general framework. Interestingly, it turns out that none of these equilibrium concepts is computationally harder than the others in any security game captured by our  framework.

\subsection{Our Results}
We start with a unified formulation of security games. In particular, we show that security games are essentially \emph{bilinear} games, in which each player's payoff has the form $\mathbf{x}^T A \mathbf{y} + \alpha \cdot \mathbf{y}$; the defender's mixed strategy $\mathbf{x}$ lies in a polytope $\P \subseteq \RR^n$ and the attacker's mixed strategy $\mathbf{y}$ is in the $n$-dimensional simplex $\Delta_n$. Interestingly, the vertices of $\P$, i.e., all the defender pure strategies, form a set system $\E$, and the defender's best response problem can be viewed as a combinatorial optimization problem over $\E$. This general framework captures most of the basic security game models in the literature, including all the deployed security systems. We show that, the set system $\E$ arising from various security domains encodes many standard combinatorial problems like bipartite matching, maximum coverage, min-cost flow, packing problems, etc. 

We are interested in solving the class of security games \emph{over} $\E$, by which we mean all security games with arbitrary payoff structures, but a fixed set $\E$ of defender pure strategies. Our main theoretical results build connections between combinatorial optimization over $\E$ and equilibrium computation for security games over $\E$. In particular, we prove that, for any set system $\E$, the following problems can be reduced to \emph{each other} in polynomial time: (0) combinatorial optimization over $\E$; (1) computing the minimax equilibrium for zero-sum security games over $\E$; (2)  computing the strong Stackelberg equilibrium for security games over $\E$; (3) computing the best or worst (for the defender) Nash equilibrium for security games over $\E$. Therefore, the hardness [polynomial solvability] of any of these problems implies the hardness [polynomial solvability] of all the others. This shows that the complexity of a security game is essentially determined by the set system $\E$. As applications of these results, we also show how to use them to easily recover and strengthen some known complexity results in the literature, as well as to resolve some open problems from previous work.  We remark that though these results are primarily theoretical, our reductions from equilibrium computation to combinatorial optimization can be practically implemented via well-known engineering approaches, e.g., the column generation technique.

\subsection{Related Work}
Several papers have examined the computational complexity of security games in particular settings. The most relevant are the following two papers:  Korzhyk et al. \cite{korzhyk2010complexity} consider the security settings where each security resource can be allocated to protect a subset of targets;  Letchford and Conitzer \cite{letchford2013}  consider security games on graphs where targets are nodes and security resources patrol along paths. They show polynomial solvability or NP-hardness under different conditions. To the best of our knowledge, there is no other work which specifically focuses on a complexity study of security games. Nevertheless, some hardness results are provided separately in different work for different models, e.g., \cite{gan2015security,Brown16a}. We note that our framework only concerns the basic security game models. There are various refinements of the basic models, e.g., the Bayesian setting, repeated setting, stochastic setting, etc. Examining the complexity of these settings is an interesting avenue for future work, but  is not the focus of this paper. 

Also related to our work is the rich literature on equilibrium computation for succinctly represented games. The most fundamental problem along this line is to compute one Nash equilibrium for a two-player normal-form game. This is proven to be PPAD-hard \cite{daskalakis2009,chen2009}. In the same setting, computing the  Nash equilibrium that maximizes one player's utility is NP-hard \cite{gilboa1989nash,conitzer2008}, but the strong Stackelberg equilibrium can be computed in polynomial time by solving linear programs \cite{conitzer2006}. Immorlica et al. \cite{Immorlica2011} consider the computation of bilinear zero-sum games, and show how to compute the minimax equilibrium when both players' action polytopes have explicit polynomial-size representations. They also reduce computing an $\epsilon$-minimax equilibrium to an additive FPTAS for the player's best response problem, using the no regret learning framework. However, they do not consider the reverse direction, namely, the reduction from best response to equilibrium computation. Garg et al. \cite{garg2011bilinear} consider bilinear general-sum games, and show that such games encode many interesting classes of games, including two-player normal-form games, two-player Bayesian games, polymatrix games, etc., hence are hard to solve in general. They propose polynomial time algorithms for computing approximate equilibria when the payoff matrices have low rank.

\section{Preliminaries}\label{sec:prem}
\subsection{Game Theory Basics}\label{sec:prem:game}
A \emph{bilinear game} is given by a pair of matrices $(A,B)$ and polytopes $(\P, \Q )$.  Given that player 1 plays $\mathbf{x} \in \P$ and player 2 plays $\mathbf{y} \in \Q$, the utilities for player 1 and 2 are $\mathbf{x}^T A \mathbf{y}$ and $\mathbf{x}^T B \mathbf{y}$ respectively. As we will show later, security games are essentially bilinear games with slightly richer structure of the form $\mathbf{x}^T A \mathbf{y} + \alpha \cdot \mathbf{y}$ and $\mathbf{x}^T B \mathbf{y} +\beta \cdot \mathbf{y}$. Note that each vertex of the polytope is a pure strategy, and a player may have exponentially many pure strategies even though her action polytope (e.g., a hypercube) can be compactly represented. Throughout this paper, we assume all the action polytopes are compact.

{\bf Nash Equilibrium (NE).} A strategy profile $(\mathbf{x} ,\mathbf{y})$ is called a Nash Equilibrium (NE), if
\begin{equation*}
\mathbf{x}^T A \mathbf{y} \geq \mathbf{x}'^T A \mathbf{y}, \, \forall \mathbf{x}' \in \P \qquad \text{and} \qquad \mathbf{x}^T B \mathbf{y} \geq \mathbf{x}^TB\mathbf{y}' , \, \forall \mathbf{y}'  \in \Q.
\end{equation*}
According to Nash's theorem, there exists at least one NE, possibly multiple NEs, in a bilinear game. As observed in \cite{garg2011bilinear}, bilinear games are general enough to capture many interesting classes of games, therefore computing an NE is hard in general bilinear games. 

{\bf Strong Stackelberg Equilibrium (SSE).} NE captures the equilibrium outcome of a simultaneous-move game. However, when one player moves before another player, the Strong Stackelberg Equilibrium \cite{Stackelberg,Stengel} serves as a more appropriate solution concept. A two-player Stackelberg game is played between a \emph{leader} and a \emph{follower}. The leader moves first, or equivalently, commits to a mixed strategy; The follower observes the leader's strategy and best responds. The leader's optimal strategy, together with the follower's best response, forms an SSE. Formally, let $\mathbf{y}_{\mathbf{x}} = \arg \max_{\mathbf{y}' \in \Q}  \mathbf{x}^T B   \mathbf{y}'$ denote the follower's best response to a leader strategy $\mathbf{x} \in \P$. A strategy profile $(\mathbf{x} ,\mathbf{y})$ is called a Strong Stackelberg Equilibrium (SSE), if
\begin{equation*}
\mathbf{x} = \arg \max_{\mathbf{x}' \in \P } \mathbf{x}'^T A \mathbf{y}_{\mathbf{x}'}, \, \qquad \text{ and } \qquad \mathbf{y} = \mathbf{y}_{\mathbf{x}}. 
\end{equation*}
Without loss of generality, we can assume $\mathbf{y}$ is a pure strategy since it is a best response to $\mathbf{x}$. 

{\bf Zero-Sum Games and the Minimax Equilibrium.}  When $A =  - B$, the bilinear game is \emph{zero-sum}. As widely known,  all standard equilibrium concepts, including the NE and SSE, are payoff-equivalent to the minimax equilibrium in zero-sum games. A strategy profile $(\mathbf{x} ,\mathbf{y})$, where $\mathbf{x} \in \P$ and $\mathbf{y} \in \Q$, is called a minimax equilibrium if
\begin{equation*}
\mathbf{x}^T A \mathbf{y} \geq \mathbf{x}'^T A \mathbf{y}, \, \forall \mathbf{x}' \in \P \qquad \text{and} \qquad \mathbf{x}^T A \mathbf{y} \leq \mathbf{x}^TA \mathbf{y}' , \, \forall \mathbf{y}'  \in \Q.
\end{equation*}If $(\mathbf{x} ,\mathbf{y})$ is a minimax equilibrium, $\mathbf{x}$ is called the player 1's \emph{maximin} strategy, $\mathbf{y}$ is called the player 2's \emph{minimax} strategy and $V = \mathbf{x}^T A \mathbf{y} = \max_{\mathbf{x}' \in \P} \min_{\mathbf{y}' \in \Q} \mathbf{x}'^TA\mathbf{y}'$ is called the \emph{value} of the game. Each zero-sum game  has a unique game value.

\subsection{Linear Optimization and Convex Decomposition via Oracles}\label{sec:prem:opt}
Any linear optimization problem can be implicitly described as $\max_{x \in \P} c^Tx$ where $\P \subseteq \RR^N$ is a polytope. 
By ``(linear) optimization over $\P$" we mean solving the problem $\max_{x \in \P} c^Tx$ for any $c \in \RR^N$. 
A \emph{membership oracle} for $\P$ is an algorithm that, on input $x_0 \in \RR^N$, correctly asserts whether $x_0$ is in $\P$ or not. 
A \emph{separation oracle} for $\P$ is an algorithm that, on input $x_0 \in \RR^N$, either asserts $x_0 \in \P$  or finds a hyperplane $a^Tx = b$ that separates $x_0$ from $\P$ in the following sense: $a^T x_0 >b$ and $a^T x \leq b$ for any $x \in \P$. The membership [separation] problem for polytope $\P$ is to compute a membership [separation] oracle for $\P$. In the following theorems, by ``polynomial running time" we mean polynomial in the dimension $N$ and the description length of the instance; Solving an LP means returning the optimal objective value as well as a vertex optimal solution. 
The following celebrated results are due to  Gr{\"o}tschel, Lov{\'a}sz and Schrijver \cite{grotschel2012}.

\begin{theorem}\label{thm:sep}
	Let $\P \subseteq \RR^N$ be a polytope. If there is a polynomial time algorithm to solve the separation problem for  $\P$, then there is a polynomial time algorithm to solve any linear program over $\P$ as well as its dual program. 
\end{theorem}

\begin{theorem}\label{thm:mem}
	(Also see \cite[P.189]{schrijver1998}) 	Let $\P \subseteq \RR^N$ be a polytope, and assume $x_0 \in \P$ be a point known to any algorithm. Then there is a polynomial time algorithm for the optimization problem over $\P$ if and only if there is a polynomial time algorithm for the membership problem for  $\P$. 
\end{theorem}

To implement a mixed strategy $x \in \P$, we need to decompose $x$ into a linear combination of pure strategies, i.e., vertices of $\P$. This can also be done efficiently given access to an efficient separation oracle. 
\begin{theorem}\label{thm:decompose}
	Let $\P \subseteq \RR^N$ be a polytope. If the separation problem for $\P$ can be solved in polynomial time, then there is a polynomial time algorithm that, on any input  $x \in \P$, computes $N+1$ vertices $x_1,...,x_{N+1}\in \P$ and convex coefficients $\lambda_1,...,\lambda_{N+1}$ such that $x = \sum_{i=1}^{N+1} \lambda_i x_i$.
\end{theorem}

\section{The Model of Security Games}\label{sec:SG}
\subsection{Strategies and Payoff Structures}
A security game is a two-player game played between a \emph{defender} and an \emph{attacker}. The defender possesses multiple security resources and aims to allocate these resources to protect $n$ \emph{targets} (e.g., physical facilities, critical locations, etc.) from the attacker's attack. We use $[n]$ to denote the set of these targets. A \emph{defender pure strategy} is a subset of targets that is protected (a.k.a., \emph{covered}) in a feasible allocation of these resources. For example, the defender may have $k(<n)$ security resources, each of which can be assigned to protect any target. In this simple example, any subset of $[n]$ with size at most $k$ is a defender pure strategy. However, in practice, there are usually resource allocation constraints, thus not all such subsets correspond to feasible allocations. We will provide more examples in Section \ref{sec:SG:comb}.    

A more convenient representation of a pure strategy, as will be used throughout the paper, is a \emph{binary} vector $\mathbf{e} \in \{0,1\}^n$,  in which the entries of value $1$ specify the covered targets. Let $\E \subseteq \{0,1\}^n$ denote the set of all defender pure strategies. Notice that $\E$ also represents a \emph{set system}. The size of $\E $ is very large, usually exponential in the number of security resources. In the example mentioned above, $|\E| = \Omega(n^k)$ which is exponential in $k$. Therefore,  computational efficiency in security games means time polynomial in $n$, \emph{not} \emph{$|\E|$}. 
A defender mixed strategy is a distribution $\mathbf{p}$ over the elements in $\E$. The attacker chooses one target to attack, thus an \emph{attacker pure strategy} is a target $i \in [n]$. We use $\mathbf{y} \in \Delta_n$ to denote an attacker mixed strategy where $y_i$ is the probability of attacking target $i$.

The payoff structure of the game is as follows: given that the attacker attacks target $i$, the defender gets a reward $r_i$ if target $i$ is covered or a cost $c_i$ if $i$ is uncovered; while the attacker gets a cost $\zeta_i$ if target $i$ is covered or a reward $\rho_i$ if $i$ is uncovered.
Both players have utility $0$ on the other $n-1$ unattacked targets. 
A crucial structure of security games is summarized in the following assumption: $r_i > c_i$ and $\rho_i > \zeta_i$ for all $i \in [n]$. That is, covering a target is strictly beneficial to the defender than uncovering it; and the attacker prefers to attack a target when it is uncovered.\footnote{In practice, the attacker can also choose to not attack. This can be incorporated into the current model by adding a dummy target. Therefore, we will not explicitly consider the case here.} We formalize the model of security games in the following definition.
\begin{definition}\label{def:SG}[{\bf Security Game}] 
	A security game $\G$ with $n$ targets is given by the following tuple $( \mathbf{r},\mathbf{c},\mathbf{\rho},\mathbf{\zeta},\E )$ and satisfies $r_i > c_i$ and $\rho_i > \zeta_i$ for all $i \in [n]$. The security game is  \emph{zero-sum} if $r_i + \zeta_i = 0$ and $c_i + \rho_i = 0$ for all $i \in [n]$.
\end{definition}
We denote a security game by $\G(\mathbf{r},\mathbf{c},\mathbf{\rho},\mathbf{\zeta},\E )$. When the game is  zero-sum, we also use  $\G(\mathbf{r},\mathbf{c},\E)$ for short. We are interested in solving security games \emph{over} $\E$, by which we mean security games with arbitrary payoff structures, but a fixed set $\E$ of defender pure strategies. 
The defender's utility, as a function of the defender pure strategy $\mathbf{e}$ and attacker pure strategy $i$, can be formally expressed as 
\begin{equation*}
U^d(\mathbf{e},i)
= r_i \cdot e_{i}+ c_i \cdot (1 - e_{i}),
\end{equation*}
where $e_{i}$ is the $i$'th entry of $\mathbf{e}$. Given a defender mixed strategy $\mathbf{p} \in \Delta_{|\E|}$ and attacker mixed strategy $\mathbf{y} \in \Delta_n$, we use $U^d(\mathbf{p},\mathbf{y})$ to denote the defender's expected utility, which can be expressed as
\begin{equation} \label{eq:defU}
\begin{array}{lll}
U^d(\mathbf{p},\mathbf{y}) &=& \sum_{\mathbf{e} \in \E} \sum_{i=1}^n p_e y_iU^d(\mathbf{e},i) \\
&=& \sum_{\mathbf{e} \in \E} \sum_{i=1}^n p_ey_i \bigg(  r_i \cdot e_{i} + c_i \cdot  (1 - e_{i}) \bigg) \\
& = &  \sum_{i=1}^n y_i \bigg(  r_i \cdot \sum_{\mathbf{e} \in \E} p_e e_{i} + c_i \cdot (1 - \sum_{\mathbf{e} \in \E} p_e e_{i}) \bigg) \\
& = &  \sum_{i=1}^n y_i \bigg(  r_i \cdot x_i + c_i \cdot [1 - x_i] \bigg)
\end{array}
\end{equation}
where 
\begin{equation}\label{eq:marginalDef}
x_i = \sum_{\mathbf{e} \in \E} p_e e_{i} \in [0,1]
\end{equation} is the \emph{marginal} coverage probability of target $i$. Let $\mathbf{x}=(x_1,...,x_n)^T$ denote the marginal probability for all targets induced by the mixed strategy $\mathbf{p}$.  Notice that the marginal probability induced by a pure strategy $\mathbf{e}$ is precisely $\mathbf{e}$ itself.\footnote{Here we assume security forces have perfect  protection effectiveness. That is, once a target is covered, regardless by one or multiple resources, it is fully protected with probability $1$. Generalization to nonperfect effectiveness is straightforward.}   Equation \eqref{eq:defU} shows that the defender's expected utility can be compactly expressed as the bilinear form $\mathbf{x}^TA\mathbf{y} + \alpha \cdot \mathbf{y}$ for some \emph{non-negative} diagonal matrix $A$, where $\mathbf{x}$ is the marginal probability induced by the defender mixed strategy.  We note that the convex hull of $\E$ forms a polytope $\P=\{ \mathbf{x}: \mathbf{x} = \sum_{\mathbf{e} \in \E} p_e \cdot \mathbf{e} , \, \forall \mathbf{p} \in \Delta_{|\E|} \}$ which consists of all the \emph{feasible} (i.e., implementable by a defender mixed strategy) marginal probabilities.\footnote{One type of security games that is not captured by this model is the network interdiction game \cite{washburn1995two,tsai2010}, in which the defender chooses edges to defend and the attacker chooses a path to attack. The task of the defender is to interdict the attacker at a certain edge. The utility functions in this game do not have the bilinear structure, and are non-convex in general.  This belongs to the more general class of succinctly represented games with non-linear payoffs, which is EXP-hard to solve even in zero-sum cases \cite{feigenbaum95}.}    

For the rest of this paper, we will simply interpret a point $\mathbf{x} \in \P$ as a mixed strategy, and instead write the defender's utility as $U^d(\mathbf{x},\mathbf{y})$. Similarly, the attacker's expected utility can be compactly represented in the following form. We note that $U^a(\mathbf{x},\mathbf{y})$ also has the bilinear form $\mathbf{x}^TB\mathbf{y} + \beta \cdot \mathbf{y}$ for some \emph{non-positive} diagonal matrix $B$.
\begin{equation}\label{eq:attU}
U^a(\mathbf{x},\mathbf{y}) = \sum_{i=1}^n y_i \bigg(    \rho_i \cdot [1-x_i] + \zeta_i \cdot x_i\bigg).
\end{equation}

\subsection{Equilibrium Concepts}\label{sec:SG:equ}
Many security games, including some deployed systems \cite{an2012,yin2012}, are modeled as zero-sum games. That is, the defender's reward [cost] is the negative of the attacker's cost [reward]. For example, in the deployed security system for patrolling proof-of-payment metro-systems \cite{yin2012}, the defender aims to catch fare evaders at metro stations. This game is naturally zero-sum: the evader's cost of paying a fine is the defender's reward of catching the evader, while the ticket price is the evader's reward and the defender's cost if failing to catch the evader.  In zero-sum games, all standard equilibrium concepts are payoff-equivalent to the  \emph{minimax equilibrium}, and our goal is to compute the minimax equilibrium in $\poly(n)$ time.

When the game is \emph{not} zero-sum, the main solution concept adopted in the literature of security games is the \emph{strong Stackelberg equilibrium} ({\bf SSE}) \cite{Stackelberg,Stengel}. In particular, the defender plays the role of the \emph{leader} and can \emph{commit} to a mixed strategy before the attacker moves. The attacker observes the defender's mixed strategy and best responds. This is motivated by the consideration that the attacker usually does surveillance before committing an attack, thus is able to observe the empirical distribution of the defender's patrolling strategy \cite{tambe2011}.   In this case, our goal is to compute the optimal mixed strategy for the defender to commit to (the attacker's best response problem is usually trivial). Notice that, the attacker is not able to observe the defender's real-time deployment (i.e., the sampled pure strategy) since he has to plan the attack before the defender's real-time pure strategy is sampled.

The strong Stackelberg equilibrium (SSE) is appropriate only when the attacker does surveillance and can indeed observe the defender's past actions. However, in many cases the attacker does little surveillance. In fact, even if the attacker intends to do surveillance, sometimes he cannot observe the defender's strategies due to limited resources and confidentiality of the defender's resource allocation (e.g., plainclothes police). In these settings, the defender cannot commit to a strategy, thus  \emph{Nash equilibrium} ({\bf NE}) serves as a more appropriate solution concept. Simultaneous-move security game models are particularly common for modeling interactions with terrorism, partially due to the fact that the defender's actions are confidential in such settings (see, e.g., \cite{major2002advanced,sandler2003terrorism,sandler2005counterterrorism,bier2007choosing}).  In networked information systems, the interaction between the defender (system protector) and attacker (malware) is usually modeled as a simultaneous-move security game as well since malwares typically do not analyze past system behaviors. The goal is usually to compute  some particular (e.g., best or worst) Nash equilibrium \cite{mavronicolas2005,mavronicolas2006}.  

\subsection{Security Games \& Combinatorial Optimization}\label{sec:SG:comb}
The main theme of this paper is to build connections between combinatorial optimization and equilibrium computation in security games. In particular, we consider the following combinatorial problem.

\begin{problem}[Defender Best Response ({\bf DBR})]\label{prob:DBR}
	For any \emph{non-negative} weight vector $\mathbf{w} \in \mathbb{R}_+^n$, compute 
	\begin{equation*}
	\mathbf{e}^*  = \arg \max_{\mathbf{e} \in \E} [\mathbf{w}\cdot \mathbf{e}].
	\end{equation*} 
	The DBR problem \emph{over} $\E$ is to compute $\arg \max_{\mathbf{e} \in \E} [\mathbf{w}\cdot \mathbf{e}]$ for any input $\mathbf{w} \in \mathbb{R}_+^n$.
\end{problem}
\noindent In other words, the DBR problem is to compute a defender pure strategy that maximizes the total weight it ``collects". 
We claim that Problem \ref{prob:DBR} is precisely the defender's best response problem to an arbitrary attacker mixed strategy. To see this, given any attacker mixed strategy $\mathbf{y}$, we have $U^d(\mathbf{x},\mathbf{y}) = \sum_{i=1}^n x_i \big( y_i [ r_i - c_i] \big) + \sum_{i=1}^n y_ic_i$ for any $\mathbf{x} \in \P$. Let $w_i = y_i [ r_i - c_i] \geq 0$, then the defender's best response to $\mathbf{y}$ is $\arg \max_{\mathbf{x} \in \P} \mathbf{x} \cdot \mathbf{w} = \arg \max_{\mathbf{e} \in \E} \mathbf{e} \cdot \mathbf{w}$. Conversely, given any $\mathbf{w} \in \RR^n_+$, it is easy to find an attacker mixed strategy $\mathbf{y} \in \Delta_n$ such that $y_i(r_i-c_i)$ is proportional to $w_i$ for all $i \in [n]$, making Problem \ref{prob:DBR} equivalent to the defender's best response to $\mathbf{y}$.   Notice that the DBR problem is a combinatorial optimization problem over the set system $\E$. The difference among various security game models  essentially lies in the structure of $\E$. Next, we illustrate how some typical DBR problems relate to standard combinatorial problems.  

{\bf Uniform Matroid.} In simple security settings, the defender has a certain number of security resources, say $k$ resources; each resource can be assigned to protect any (one) target -- i.e., no allocation constraints. Therefore, any subset of $[n]$ of size at most $k$ is a defender pure strategy. In this case, $\E$ is a uniform matroid and the DBR problem is simply to find the largest $k$ weights. 
The deployed security system for LAX airport checkpoint placement -- one of the earliest applications of security games -- is captured by this model \cite{pita2008}.     

{\bf Bipartite Matching.} A natural generalization of the uniform matroid case is that the resource allocation is constrained. In this case, the defender has $k$ heterogeneous resources, and each resource can only be assigned to some particular targets associated with that resource. This naturally models several types of scheduling constraints in practice. For example, due to geographic constraints, policemen from a certain police station can only patrol the area around that station. Also, different types of security forces specialize in protecting different types of targets. Here, feasible assignments can be modeled as edges of a bipartite graph with security resources on one side and targets on another side. A defender pure strategy corresponds to a bipartite matching, and the DBR problem is to compute the maximum weighted bipartite matching.

{\bf Coverage Problems.} In some domains, one security resource can cover several targets. One deployed real-world example is the federal air marshal scheduling problem, where one air marshal is assigned to protect several flights, but constrained on that the arrival destination of any former flight should be the starting point of the next flight \cite{Tsai09a}. In other words, each security resource (i.e., air marshal) can protect a subset of targets (i.e., flights). Therefore, each pure strategy is the union of targets covered by each security resource. The DBR problem in this case is the maximum weighted coverage problem. Another natural example is to protect targets distributed on the plane and each security guard can cover a region of certain size. The DBR problem here is a 2-dimensional geometric maximum coverage problem.\footnote{For more information about geometric coverage, see the thesis  \cite{leeuwen2009} and references therein.} Other examples include patrolling on a graph (e.g., street map or network systems) in which a patroller at a node can protect all adjacent edges or a patroller on an edge can protect its two end nodes. The DBR problem here is the vertex or edge coverage problem. 

{\bf Min-Cost Flow.} Many security games are played out in \emph{spatio-temporal} spaces. For example, the deployed security system in \cite{fang2013optimal} helps to schedule the patrol boats of the US Coast Guard to protect the (moving) Staten Island ferries every day. Wildlife protection is another example \cite{Fang16a}. One common way to handle such settings is to discretize the space and build a 2-D -- spatial and temporal dimension --  grid, and patrol the discrete \emph{(space, time)} points (see Figure \ref{fig:flow}). However, starting from a position at time $t$, the positions that a security resource can possibly reach at time $t+1$ are restricted due to various constraints like speed limit, terrain barriers, etc. For example, the move highlighted by red in Figure \ref{fig:flow} may be infeasible since a patroller can not move that far within a small time period due to speed limit, while the blue-colored moves are feasible. This can be modeled by adding edges between time layers to indicate feasible moves. 
The patrolling schedule for each security resource corresponds to a path across all time layers, which specifies the position this resource covers at each time point (see the blue path in Figure \ref{fig:flow}). The DBR problem is, for any given non-negative weights at each \emph{(space, time)} point, computing $k$ paths for the $k$ resources to maximize the total weights they cover. This can be converted to a min-cost max flow problem (with negative costs). 
\begin{figure}
	\vspace{-10mm}
	\centering
	\includegraphics[bb=70bp 180bp 750bp 600bp,clip,scale=0.35]{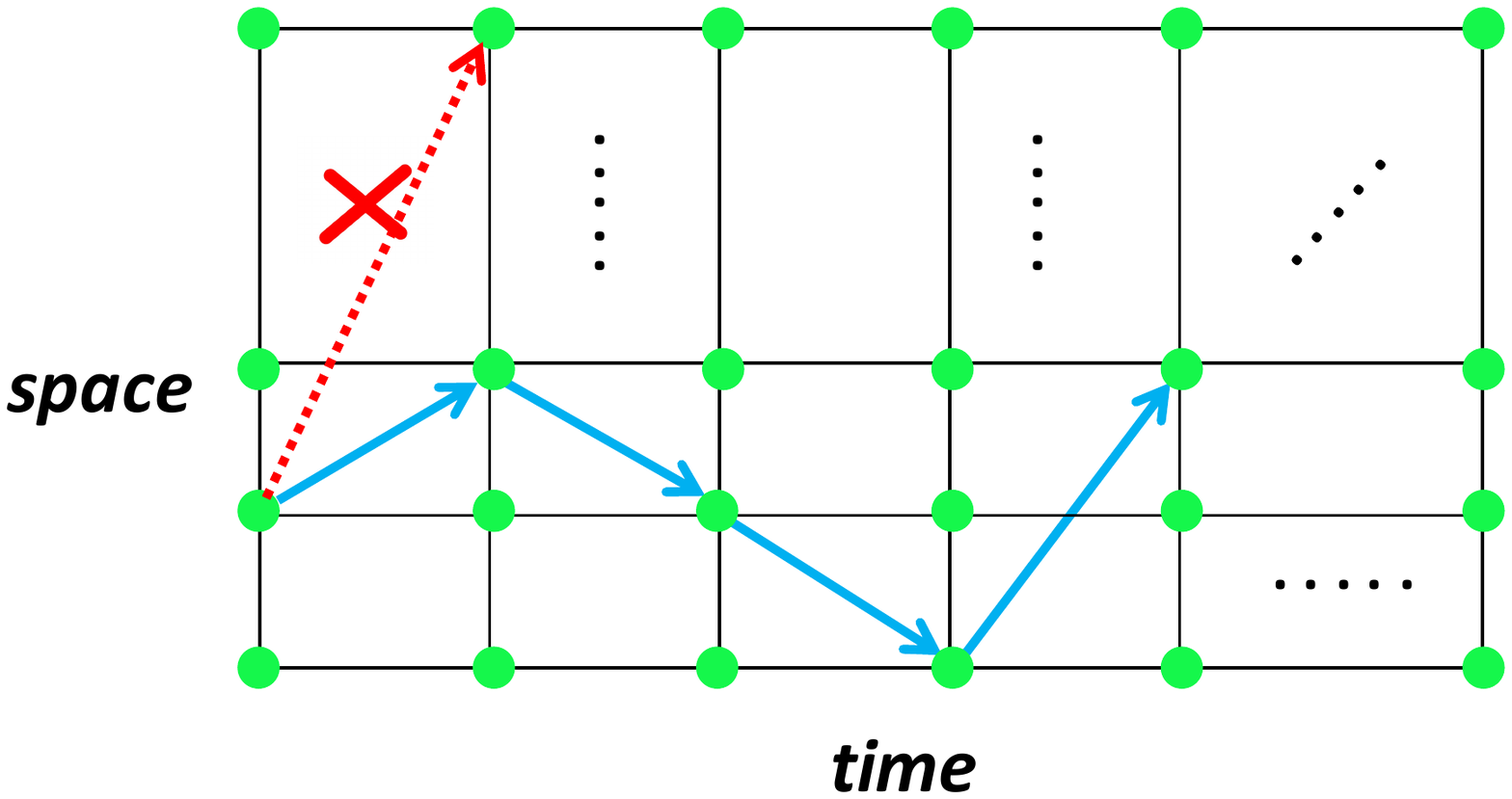}
	\caption{\label{fig:flow} Feasible (blue) and infeasible (red) moves in  spatio-temporal spaces.}
\end{figure}

{\bf Packing Problems.} Our last example is motivated by recent work on optimizing the allocation of security resources for passenger screening for the Transportation Security Administration
(TSA) \cite{Brown16a}. Consider an airport with $n$ flights and flight $i$ has $m_i$ passengers. The TSA has several screening \emph{tools}, e.g., x-ray, walk-through metal detector, chemicals, etc.,  and each screening tool has a capacity of the maximum number of passengers it can check. In contrast to the case of coverage problems where each resource can protect several targets, here several tools are required to screen one passenger. More precisely, each passenger is screened by a screening \emph{team} which is a combination of several screening tools.  
By protecting the flight from a certain passenger we mean a passenger of that flight is screened, and identified as an attacker if he is, by a screening team.\footnote{ In \cite{Brown16a}, each screening team has an effectiveness factor denoting the probability a team can identify an attacker.  The setting here is slightly simplified with perfect effectiveness factor $1$. Nevertheless, it still captures the core difficulty of the problem.} The DBR problem is to, given non-negative weight $w_i$  for passengers in flight $i$, allocate as many passengers as possible to teams for screening, subject to each screening tool's capacity constraint, so that the total weight of  screened passengers are maximized. This is a very general packing problem. In fact, the reader may easily verify that even when $w_i = 1$ and $m_i=1$ for any flight $i$, the problem encodes the independent set problem. 

\begin{remark}
	We note that the list of examples above is by no means complete -- there are various other security settings with other combinatorial structures. For example, there is also study on budget constraints for acquiring security resources (e.g., \cite{bhattacharya2011}), which induces the budgeted version of the above combinatorial problems. In fact, real domains are usually more complicated with various types of constraints, involving intersections of these combinatorial structures. 
\end{remark}

\section{Solving Zero-Sum Security Games is a Combinatorial Problem}\label{sec:zero}
In this section, we focus on zero-sum security games. Recall that we use $n$ to denote the total number of targets, and reward $r_i$ [cost $c_i$] to denote the defender's utility of covering [uncovering] target $i$ when it is attacked. The defender seeks to maximize her utility while the attacker seeks to minimize the defender's utility. We are interested in computing the minimax equilibrium for security games over $\E$ in $\poly(n)$ time. More specifically, we seek to understand how the computational complexity of the minimax equilibrium relates to the complexity of the DBR problem. By convention, we sometimes call an algorithm for solving the DBR problem a \emph{DBR oracle}. We prove the following equivalence theorem.
\begin{theorem}\label{thm:main}
	There is a $\poly(n)$ time algorithm to compute
	the minimax equilibrium for zero-sum security games  over $\E$, \emph{if and only if} there is a $\poly(n)$ time algorithm to solve the $DBR$ problem over $\E$.
\end{theorem}

Since security games are essentially bilinear games, it is not surprising that the minimax equilibrium can be reduced an the DBR problem, and the reduction  follows a standard primal-dual argument. 
What is interesting, however, is the reverse direction -- i.e., solving the DBR problem is no harder than solving zero-sum security games. The minimax equilibrium, as an optimization problem,  has a special objective function. It is not clear that such a special objective can be as hard as optimizing an arbitrary non-negative objective over $\E$.  Moreover, security games are very special bilinear games: (i) the attacker's mixed strategy set is a simplex; (ii) the defender's  payoff matrix is non-negative and diagonal; (iii) the attacker's  payoff matrix is non-positive and diagonal. Thus unsurprisingly, there have been attempts in the literature on solving security games without going through the DBR problem. Indeed, various other techniques have been employed to tackle security games, e.g., generalized Birkhoff-von Neumann theorem \cite{budish2013}, techniques from convex and non-convex optimization, etc.  However, the message conveyed in Theorem \ref{thm:main} is that, to solve security games in polynomial time, it is not only sufficient, but also necessary, to solve the DBR problem in polynomial time. We now provide a formal proof of Theorem \ref{thm:main}. 

\begin{figure}
	\centering
	\begin{minipage}{.5\textwidth}
		\centering
		\begin{lp*}
			\maxis{u}
			\sts
			\qcon{x_i r_i + (1-x_i)c_i \geq u}{i =1,2,...,n}
			\con{ \sum_{\mathbf{e} \in \E} p_e \cdot \mathbf{e}  = \mathbf{x}}
			\con{\sum_{\mathbf{e} \in \E} p_e = 1}
			\qcon{p_e \geq 0}{\mathbf{e} \in \E}
		\end{lp*}
		\vspace{-8pt}
		\captionof{figure}{The defender's maximin strategy\\ (variables: $x_i$, $p_e$ and $u$)}
		\label{lp:ZeroSumSSE}
	\end{minipage}%
	\begin{minipage}{.5\textwidth}
		\centering
		\begin{lp*}
			\minis{\sum_{i=1}^n c_i y_i + r}
			\sts
			\qcon{r - \mathbf{e} \cdot \mathbf{w}  \geq 0}{\mathbf{e} \in \E}
			\qcon{ (r_i - c_i)y_i = w_i }{ i  = 1,2, ...,n}
			\con{\sum_{i =1}^n y_i = 1}
			\qcon{y_i \geq 0}{i=1,2,...,n}
		\end{lp*}
		\vspace{-8pt}
		\captionof{figure}{The attacker's minimax strategy \\  (variables: $y_i$, $w_i$ and $r$)}
		\label{lp:ZeroSumDual}
	\end{minipage}%
\end{figure}

\subsection*{Reducing Minimax Equilibrium to DBR}
Figure \ref{lp:ZeroSumSSE} and \ref{lp:ZeroSumDual} exhibit the linear programs for computing the defender's maximin strategy and the attacker's minimax strategy. They are primal and dual of each other. We observe that the linear program in Figure \eqref{lp:ZeroSumDual} admits a $poly(n)$ time separation oracle, if the DBR problem can be solved in $\poly(n)$ time.  In particular, given any $r$ and $w_i, y_i$ for $i\in [n]$, we can explicitly check the last three constraints in $poly(n)$ time. To check whether they satisfy the first constraint, we run the DBR oracle which outputs a pure strategy $\mathbf{e}^*$. If   $r < \mathbf{e}^* \cdot \mathbf{w} $, then the first constraint is violated at $\mathbf{e} = \mathbf{e}^*$. Otherwise, $r \geq \mathbf{e}^* \cdot \mathbf{w} = \max_{\mathbf{e} \in \E} \geq  \mathbf{e} \cdot \mathbf{w}$ for any $\mathbf{e} \in \E$. Therefore, the first constraint is satisfied. To sum up, the linear program in Figure  \eqref{lp:ZeroSumDual}  admits a $poly(n)$ time separation oracle. By Theorem \ref{thm:sep}, both  LPs  can be solved in $poly(n)$ time. We note that the (efficiently) computed optimal solution of the LP in Figure \eqref{lp:ZeroSumSSE} will have a support of $poly(n)$ size. In fact, Theorem \eqref{thm:decompose} guarantees that there exists a defender maximin strategy that mixes over at most $n+1$ pure strategies, and it can be computed efficiently.

\subsection*{Reducing DBR to Minimax Equilibrium}	 
Recall that all the feasible marginal coverage probabilities form the polytope $\P$, and the vertices of $\P$ form the set $\E$ of defender pure strategies.  We start with a simple fact of linear programing, which states that any linear program achieves optimality at some vertex of its feasible region (if non-empty). 
\begin{fact}\label{fact}
	$\max_{\mathbf{e} \in \E} [\mathbf{w}\cdot \mathbf{e}] = \max_{\mathbf{x} \in \mathcal{P}} [\mathbf{w}\cdot \mathbf{x} ]$, where $\P$ is the convex hull of $\E$. 
\end{fact}
By Fact \ref{fact}, solving the DBR problem is equivalent to solving the linear program $ \max_{\mathbf{x} \in \mathcal{P}} [\mathbf{w}\cdot \mathbf{x} ]$.  The main technical step of our proof is to reduce the membership checking problem for polytope $\mathcal{P}$ to the computation of minimax equilibrium. Therefore,  computing the minimax equilibrium will allow us to find a $\poly(n)$ time membership oracle for $\mathcal{P}$.   By the polynomial-time equivalence between membership checking and optimization (Theorem \ref{thm:mem}), we can conclude that the DBR problem can also be solved in $\poly(n)$ time.
Unfortunately, it turns out that membership checking for polytope $\mathcal{P}$ cannot be easily reduced to the computation of minimax equilibrium because some $\mathbf{x} \in \P$ are entry-wise dominated by (i.e., entry-wise \emph{less} than) other $\mathbf{x}' \in \P$, so that $\mathbf{x}$ can never be a defender equilibrium strategy.  To overcome this barrier, we first relax the polytope $\mathcal{P}$ and work on a broader set $\hat{\P}$ of marginal probabilities. We then make use of the non-negativity of weights  in the DBR problem to transfer linear  optimization over $\hat{\P}$ back to linear optimization over $\P$.

We first introduce some notations. We call $\hat{\mathbf{e}} \in \{0,1\}^n$ a \emph{sub pure strategy} (of $\mathbf{e}$) if there exists an $\mathbf{e} \in \E$ such that $\hat{e}_i \leq e_{i}$ for all $i \in [n]$. Equivalently, the covered targets of a sub pure strategy $\hat{\mathbf{e}}$ is a subset of the covered targets of some real pure strategy $\mathbf{e}$. For example, the all zero vector $\mathbf{0}$ is  always a sub pure strategy; any pure strategy $\mathbf{e} \in \E$ is also a sub pure strategy (of itself). Notice that sub pure strategies are not necessarily feasible. For example, in the air marshal scheduling problem, a feasible schedule for an air marshal has to be a round trip, not a one-way flight. We now define the \emph{relaxed} set of pure strategies 
\begin{equation*}
\hat{\E} = \{   \text{ all $\mathbf{e} \in \E$ and all their sub pure strategies } \}
\end{equation*}
to be the set of $\E$ augmented with all sub pure strategies. The following lemma shows that, when considering the DBR problem, relaxing $\E$ to $\hat{\E}$ is without loss.
\begin{lemma}\label{lem:relaxE}
	$\max_{\mathbf{e} \in \E} [\mathbf{e} \cdot \mathbf{w}] = \max_{\hat{\mathbf{e}} \in \hat{\E}} [\hat{\mathbf{e}}\cdot \mathbf{w} ]$ for any $\mathbf{w} \in \RR_+^n.$
\end{lemma}
\begin{proof}
	Since $\E \subseteq \hat{\E}$, we have $\max_{\mathbf{e} \in \E} [\mathbf{e} \cdot \mathbf{w}] \leq \max_{\hat{\mathbf{e}} \in \hat{\E}} [\hat{\mathbf{e}}\cdot \mathbf{w} ]$. To prove another direction, let $\hat{\mathbf{e}}^* = \arg \max_{\hat{\mathbf{e}} \in \hat{\E}} [\hat{\mathbf{e}}\cdot \mathbf{w} ]$. By definition, there exists $\mathbf{e}^* \in \E$ such that $\hat{\mathbf{e}}^*$ is a sub pure strategy of $\mathbf{e}^*$. Since $\mathbf{w} \in \RR_+^n$ is non-negative, we must have $\mathbf{e}^* \cdot \mathbf{w} \geq \hat{\mathbf{e}}^* \cdot \mathbf{w}$. Therefore, 
	\begin{equation*}
	\max_{\mathbf{e} \in \E} [ \mathbf{e} \cdot \mathbf{w}] \geq \mathbf{e}^* \cdot  \mathbf{w} \geq \hat{\mathbf{e}}^*  \cdot  \mathbf{w} = \max_{\hat{\mathbf{e}} \in \hat{\E}} [\hat{\mathbf{e}}\cdot \mathbf{w} ],		
	\end{equation*}
	which concludes the proof. 
\end{proof}
We now define the \emph{relaxed} polytope of marginal probabilities as follows.
\begin{equation*}
\hat{\mathcal{P}} = \{ \mathbf{x}: \mathbf{x} = \sum_{\hat{\mathbf{e}} \in \hat{\E}} p_{\hat{e}} \cdot \hat{\mathbf{e}} , \, \forall \mathbf{p} \in \Delta_{|\hat{\E}|} \}.
\end{equation*}
We next show that the DBR problem over $\E$ can be reduced to linear optimization (or equivalently, membership checking, due to Theorem \ref{thm:mem}) of this relaxed polytope $\hat{\P}$. It is easy to show that  the optimal objective value of the DBR problem is equal to $\max_{\mathbf{x} \in \hat{\P}} [\mathbf{x}\cdot \mathbf{w} ]$. However, the problem is that the optimal (vertex) solution to $\max_{\mathbf{x} \in \hat{\P}} [\mathbf{x}\cdot \mathbf{w} ]$, i.e., some $\hat{\mathbf{e}}^* \in \hat{\E}$, is a sub pure strategy, which may not be a feasible real pure strategy (since $\mathbf{w}$ may have zero-valued entries), but an algorithm for the DBR problem must return a feasible optimal solution. Nevertheless, this can be handled by properly regularizing the linear program $\max_{\mathbf{x} \in \hat{\P}} [\mathbf{x}\cdot \mathbf{w} ]$ to make sure that its optimal vertex solution is always feasible for the DBR problem over $\E$.
\begin{lemma}\label{lem:reduce1}
	The DBR problem over $\E$ reduces to membership checking for polytope $\hat{\P} $ in $\poly(n)$ time.
\end{lemma}
\begin{proof}
	By Lemma \ref{lem:relaxE} and Fact \ref{fact}, we have $\max_{\mathbf{e} \in \E} [\mathbf{e} \cdot \mathbf{w} ] = \max_{\hat{\mathbf{e}} \in \hat{\E}} [\hat{\mathbf{e}} \cdot \mathbf{w} ] =  \max_{\mathbf{x} \in \hat{\P}} [\mathbf{x}\cdot \mathbf{w} ]$. With the access to an efficient membership oracle for $\hat{\P}$, we can solve the optimization problem $\max_{\mathbf{x} \in \hat{\P}} [\mathbf{x}\cdot \mathbf{w} ]$ and compute a vertex optimal solution, i.e., some $\hat{\mathbf{e}}^* \in \hat{\E}$,  in $poly(n)$ time by Theorem \ref{thm:mem} (note that $\hat{\P}$ has a trivial feasible point $\mathbf{0}$). If $\hat{\mathbf{e}}^*$ is also in $\E$, then we know that $\hat{\mathbf{e}}^*$ is also optimal for the DBR problem with weight $\mathbf{w}$.   However, $\hat{\mathbf{e}}^*$ may not be in $\E$ due to possible zero-valued entries in $\mathbf{w}$. We now show how to properly regularize the linear program  $\max_{\mathbf{x} \in \hat{\P}} [\mathbf{x}\cdot \mathbf{w} ]$ to guarantee that the returned optimal vertex solution is always a feasible pure strategy, and meanwhile, is still optimal for the objective $\mathbf{x}\cdot \mathbf{w}$. 
	
	Let polynomial $q(n)$ be the bit complexity of $\mathbf{w}$, and let $\eps = 2^{-q(n)}/2n$ which is still of polynomial bit length. Now consider the regularized optimization problem $\max_{\mathbf{x} \in \hat{\P}} [\mathbf{x}\cdot \mathbf{w} + \eps ( \mathbf{x} \cdot \mathbf{1} )]$. First observe that any vertex optimal solution to the regularized optimization problem will be a real pure strategy. Otherwise, substituting a sub pure strategy by the corresponding real pure strategy will not decrease the first term $\mathbf{x}\cdot \mathbf{w}$ but will strictly increase the second term $\eps (\mathbf{1} \cdot \mathbf{x})$, thus strictly increases the whole objective value. We now claim that if a pure strategy $\mathbf{e}^*$ is optimal for the regularized linear objective $\mathbf{x}\cdot \mathbf{w} + \eps ( \mathbf{x} \cdot \mathbf{1} )$ over $\hat{\P}$, it is also optimal for the DBR problem over $\E$ with weight $\mathbf{w}$. Suppose toward contradiction that $\mathbf{e}^* \in \E$ is optimal for the regularized optimization problem but is not optimal for the DBR problem. Then there exists a feasible pure strategy $\mathbf{e}'$ such that  $\mathbf{e}' \cdot \mathbf{w} > \mathbf{e}^* \cdot \mathbf{w}$. In fact, since $\mathbf{e}^*, \mathbf{e}'$ are 0-1 vectors and $\mathbf{w}$ has bit complexity $q(n)$, we have $\mathbf{e}' \cdot \mathbf{w} - \mathbf{e}^* \cdot \mathbf{w} \geq 2^{-q(n)}$. Therefore,
	\begin{align*}
	\mathbf{e}^* \cdot \mathbf{w} + \eps (\mathbf{1} \cdot \mathbf{e}^* ) &\leq \mathbf{e}' \cdot \mathbf{w}- 2^{-q(n)} + \eps n \\
	&= \mathbf{e}' \cdot \mathbf{w}- 2^{-q(n)}/2 \\
	&<  \mathbf{e}' \cdot \mathbf{w} + \eps (\mathbf{1} \cdot \mathbf{e}' ), 
	\end{align*} 
	which contradicts the fact that $\mathbf{e}^*$ is optimal for the regularized optimization problem. Therefore, $\mathbf{e}^*$ must be optimal for the DBR problem. Moreover, the regularized optimization problem can be solved in polynomial time with access to a $poly(n)$ time membership oracle for $\hat{\P}$. This proves the lemma. 
\end{proof}

We are now ready to complete the reduction by reducing membership checking for polytope $\hat{\mathcal{P}}$ to solving zero-sum security games over $\E$.  Before that, we first show the  following down-monotone property of the relaxed polytope $\hat{\P}$, which turns out to be crucial for the reduction. 
\begin{lemma}\label{lem:down}
	Polytope $\hat{\mathcal{P}}$ is down-monotone in $\RR_+^n$. That is, for any $\mathbf{x} \in \hat{\mathcal{P}}$ and $\mathbf{0} \leq \mathbf{x}' \leq \mathbf{x}$ (entry-wise), we have $\mathbf{x}' \in \hat{\P}$. 
\end{lemma}
\begin{proof}
	Let $\mathbf{x} \in \hat{\mathcal{P}}$ and $\mathbf{x}'$ satisfy $x'_i \leq x_i$ for all $i\in[n]$. We show that $\mathbf{x}' \in  \hat{\mathcal{P}}$. In particular, starting from a convex decomposition for $\mathbf{x}$, we will construct a convex decomposition for  $\mathbf{x}'$ over elements in $\hat{E}$.
	
	Since $\mathbf{x} \in \hat{\mathcal{P}}$, there exists $\mathbf{p} \in \Delta_{|\hat{\E}|}$, such that 
	\begin{equation}\label{eq:hatEdecompose}
	\mathbf{x} = \sum_{\hat{\mathbf{e}}\in \hat{\E}} p_{\hat{e}} \cdot \hat{\mathbf{e}}.
	\end{equation} For all the $\hat{\mathbf{e}}$ satisfying $\hat{e}_1 = 1$, let $\hat{\mathbf{e}}' \in \hat{E}$ be the sub pure strategy, which is the same as $\hat{\mathbf{e}}$ \emph{except} that $\hat{e}'_1 = 0$. Then we take the following operations: change the coefficient of $\hat{\mathbf{e}}$ in Equation \eqref{eq:hatEdecompose} to $p_{\hat{e}} \cdot \frac{x_1'}{x_1}$ and change the coefficient of $\hat{\mathbf{e}}'$ in Equation \eqref{eq:hatEdecompose} to $p_{\hat{e}'} + p_{\hat{e}} \cdot (1-\frac{x_1'}{x_1})$. It is easy to check that these new coefficients still sum up to $1$ and are non-negative. Therefore, Equation \eqref{eq:hatEdecompose} with the new coefficients is still a convex combination over elements in $\hat{E}$, which gives us a new vector $\mathbf{x}'' \in \hat{E}$. Moreover, $\mathbf{x}''$ satisfies $x''_1= x'_1$ and  $x''_i = x_i$ for all $i \geq 2$. Continuing these operations at all the other entries of  $\mathbf{x}''$, we will reach a convex decomposition of $\mathbf{x}'$. Therefore, $\mathbf{x}' \in \hat{E}$.		
	
\end{proof}

\begin{lemma}\label{lem:MembershipToLP}
	Membership checking for polytope $\hat{\mathcal{P}}$ reduces in $\poly(n)$ time to computing the \emph{game value} of zero-sum security games over $\E$. 
\end{lemma}
\begin{proof}
	Given any $\mathbf{x} \in \RR^n$, we show how to check whether $\mathbf{x} \in \hat{\P}$ by solving properly constructed zero-sum security games. We assume  $\mathbf{x}$ lies in $\mathbb{R}_+^n$ since otherwise $\mathbf{x} \not \in \hat{\mathcal{P}}$. Given any $ \mathbf{x} \in \mathbb{R}_+^n$, we construct the following security game instance:
	\begin{itemize}
		\item For any $i \in [n]$ such that $x_i = 0$, let $c_i = 1$ and $r_i = 2$. Therefore, $x_i r_i + (1-x_i)c_i = 1$.
		
		\item For any $i\in[n]$ such that $x_i > 0$, let $c_i = 0$ and $r_i = \frac{1}{x_i}$. Therefore, $x_i r_i + (1-x_i)c_i = 1$. 
	\end{itemize}
	
	Note that in both cases, the condition $r_i > c_i$ is satisfied. We claim that the defender's optimal utility, i.e., the game value, in the above zero-sum security  game is at least $1$ \emph{if and only if} $\mathbf{x} \in \hat{\mathcal{P}}$.
	
	$\Rightarrow$ direction: Let $(\mathbf{p}^*,\mathbf{x}^*,u^*)$ be the optimal solution to the LP in Figure \ref{lp:ZeroSumSSE}. If the defender's optimal utility is at least $1$, i.e., $u^* \geq 1$. We have \begin{equation*}
	x_i^*r_i + (1-x_i^*)c_i \geq u^* \geq 1 = x_ir_i + (1-x_i)c_i, \, \forall i \in [n].
	\end{equation*} 
	Since $r_i > c_i$, this induces $x_i^* \geq x_i$ for any $i$. Notice that the optimal (feasible) marginal probability $\mathbf{x}^*$ is in $ \mathcal{P}$ by definition and $\mathcal{P} \subseteq \hat{\mathcal{P}}$, therefore  $\mathbf{x}^* \in \hat{\mathcal{P}}$. However, $x_i \leq x_i^*$ for any $i$ and $\hat{\mathcal{P}}$ is down monotone, so $\mathbf{x} \in \hat{\mathcal{P}}$, as desired. 
	
	$\Leftarrow$ direction: Let $\mathbf{x} \in \hat{\mathcal{P}}$. So there exists $\mathbf{p} \in \Delta_{|\hat{\E}|}$ such that 
	\begin{equation}\label{eq:subpure}
	\mathbf{x} = \sum_{\hat{\mathbf{e}} \in \hat{\E}} p_{\hat{e}} \cdot  \hat{\mathbf{e}}.
	\end{equation} Notice that each $\hat{\mathbf{e}}$ is a sub pure strategy of some real pure strategy $\mathbf{e}(\hat{\mathbf{e}})\in \E$. By substituting the $\hat{\mathbf{e}}$ in Equation \eqref{eq:subpure} by   $\mathbf{e}(\hat{\mathbf{e}})$, we have 
	\begin{equation*}
	\mathbf{x} = \sum_{\hat{\mathbf{e}} \in \hat{\E}} p_{\hat{e}} \cdot  \hat{\mathbf{e}} \leq \sum_{\hat{\mathbf{e}} \in \hat{\E}} p_{\hat{e}}  \cdot \mathbf{e}(\hat{\mathbf{e}}).
	\end{equation*} 
	Since $\mathbf{x}' = \sum_{\hat{\mathbf{e}} \in \hat{\E}} p_{\hat{e}}  \cdot \mathbf{e}(\hat{\mathbf{e}})$ is a convex combination of pure strategies in $\E$, we have $\mathbf{x}' \in \mathcal{P}$. Moreover, by playing the mixed strategy with marginal probability $\mathbf{x}'$, the defender's expected utility is at least $1$ since $x_i' \geq x_i$ for any $i \in [n]$. As a result, the optimal defender utility in this zero-sum game is at least $1$, completing the proof.     
\end{proof}   
Lemma \ref{lem:reduce1} and \ref{lem:MembershipToLP} forms a reduction from the DBR problem to the mimimax equilibrium. This together with the reduction from the minimax equilibrum to the DBR problem conclude the proof of Theorem \ref{thm:main}.

\subsection{When Gaming is Easier than Best Response}
Security games are special bilinear games. Theorem \ref{thm:main} shows that the particular problem of computing the minimax equilibrium for zero-sum security games is as hard as the general DBR problem. A natural question is whether there are instances such that the hardness of gaming and best response are strictly separated. Notice, however, that the best response problem is essentially more general, thus no easier, than solving the game. Therefore, the question really is, whether there are instances where gaming is easier than best response. The following proposition answers this in the affirmative for security games. Notice that this proposition does not contradict Theorem \ref{thm:main}, since the constructed security game instances here have further restricted payoffs, which make solving the game easy but still maintain the hardness of the best response.   

\begin{proposition}
	There exist zero-sum security games such that the minimax equilibrium can be computed in polynomial time but the DBR problem is NP-hard. 
\end{proposition}
\begin{proof}
	Consider a security game played on a complete graph $K_n$. Each \emph{edge} is a target. The defender has $k (<n)$ security resources and each resource can patrol a \emph{vertex}, by which the $n-1$ adjacent edges of this vertex are covered. An edge is covered if at least one of its end vertexes is patrolled. Given that the attacker attacks any edge $e$, the defender gets utility $1$ if it is protected and utility $0$ otherwise. The attacker seeks to minimize the defender's utility. The DBR problem for this security game is, given weight $w_e \geq 0$ for any edge $e \in K_n$, finding $k$ vertexes that maximize the total edge weight they cover. This is NP-hard by a trivial reduction from vertex cover. However, by symmetry, it is not hard to check that uniformly randomly sampling $k$ vertexes from $K_n$ to patrol is a minimax equilibrium.
\end{proof}

\section{General-sum Security Games}\label{sec:general}
In this section, we consider general-sum security games. Recall that such a game $\G$ is given by a tuple $(\mathbf{r},\mathbf{c},\mathbf{\rho},\mathbf{\zeta},\E)$ where $r_i$ [$  c_i$] is the defender's reward [cost] and $\rho_i$ [$\zeta_i$] is the attacker's reward [cost], if target $i$ is attacked. We consider the computation of the two mostly adopted equilibrium concepts in the security game literature, namely, the strong Stackelberg equilibrium (SSE) and Nash equilibrium (NE). For each equilibrium concept, we prove analogous equivalence theorem as the zero-sum case. 
\begin{theorem}\label{thm:SSE}
	There is a $\poly(n)$ time algorithm to compute
	the strong Stackelberg equilibrium for security games  over $\E$, \emph{if and only if} there is a $\poly(n)$ time algorithm to solve the $DBR$ problem over $\E$.
\end{theorem}
\begin{proof}
	The ``only if'' direction follows from Theorem \ref{thm:main}, the fact that the minimax equilibrium is payoff-equivalent to the strong Stackelberg equilibrium (SSE) in zero-sum games, and that zero-sum games are special cases of general-sum games. We only prove the ``if" direction. 
	Recalling the definition of SSE in Section \ref{sec:prem:game}, without loss of generality, we can assume the attacker always plays a pure strategy since he moves after the defender. As observed in \cite{conitzer2006}, to compute the defender's optimal mixed strategy, we can enumerate all the possibilities of the attacker's best response choices. In particular, constrained on that the attacker's best response is target $k$, the defender's optimal strategy can be computed by the following linear program (denoted as $LP_k$):
	\begin{lp*}\label{lp:SSE}
		\maxi{x_k r_k + (1-x_k)c_k }
		\st
		\qcon{(1-x_k) \rho_k + x_k \zeta_k \geq (1-x_i) \rho_i + x_i \zeta_i  }{i \not = k}
		\con{ \sum_{\mathbf{e} \in \E} p_e \cdot \mathbf{e}  = \mathbf{x}}
		\con{\sum_{\mathbf{e} \in \E} p_e = 1}
		\qcon{p_e \geq 0}{\mathbf{e} \in \E}
	\end{lp*}where the first set of constrains is to guarantee that the attacker is indeed incentivized to attack target $k$. It is not hard to check that the dual program of $LP_k$ admits a $\poly(n)$ time separation oracle if the DBR problem can be solved in $\poly(n)$ time. The SSE can be computed by solving $LP_1,...,LP_n$ and then picking the optimal defender mixed strategy (and corresponding attacker best response) of the LP with the largest objective value.  
\end{proof}

We now turn to the computation of Nash equilibria. As widely known in the literature of algorithmic game theory,  computing one Nash equilibrium for two-player normal-form games is PPAD-hard \cite{daskalakis2009,chen2009}, and is only harder for general bilinear games \cite{garg2011bilinear}. Interestingly, it turns out that computing one Nash equilibrium in security games is  relatively easy. This is due to the following characterization of Nash equilibria in security games by \cite{korzhyk2011}. 
\begin{lemma}\label{lem:NEset} \cite{korzhyk2011}
	Consider a security game $\G (\mathbf{r},\mathbf{c},\mathbf{\rho},\mathbf{\zeta},\E)$. Let $\bar{\G}(-\mathbf{\zeta}, -\mathbf{\rho},\mathbf{\rho},\mathbf{\zeta},\E)$ be the corresponding zero-sum security game by re-setting the defender's utilities. Then $(\mathbf{x},\mathbf{y})$ is a Nash equilibrium of $\G$ if and only if $(\mathbf{x},f(\mathbf{y}))$ is a minimax equilibrium of the zero-sum game $\bar{\G}$, where the one-to-one transform function $f:\RR^n \to \RR^n$ is defined as follows:
	\begin{equation}\label{fn:trans}
	f_i(\mathbf{y}) = \frac{1}{\lambda} \frac{r_i - c_i}{\rho_i - \zeta_i} y_i, \, \, \forall i \in [n], \, \, \text{ where } \lambda = \sum_{i=1}^n  \frac{r_i - c_i}{\rho_i - \zeta_i} y_i\text{ is the normalization factor}.
	\end{equation}
	Moreover, Nash equilibria of $\G$ are interchangeable. That is, if $(\mathbf{x},\mathbf{y})$ and $(\mathbf{x}',\mathbf{y}')$ are both Nash equilibria, so are $(\mathbf{x},\mathbf{y}')$ and $(\mathbf{x}',\mathbf{y})$. The attacker derives the same utility in any Nash equilibrium of $\G$. 
\end{lemma}
\noindent Notice that the transform function defined in Equation \eqref{fn:trans} is \emph{non-linear} due to the normalization factor. We sketch the intuition of Lemma \ref{lem:NEset} here, while refer the reader to \cite{korzhyk2011} for a formal proof. Note that the mapping $f$ only re-weights those \emph{non-zero} $y_i$'s whose indexes correspond to attacker best responses, so $f(\mathbf{y})$ is also a best response to $\mathbf{x}$ in $\G$, thus also in $\bar{\G}$, since the defender strategy and \emph{attacker} payoff structure in both games are the same. On the other hand, $\mathbf{x}$ is a best response to $\mathbf{y}$ in $\G$. From $\G$ to $\bar{\G}$, the defender's utility on each target is changed. The idea here is to properly rescale the attacker's attacking probability to compensate for the defender's utility change so that the defender's best response does not change. The transform function $f$ exactly does this. The interchangeability follows from the interchangeability of minimax equilibria of $\bar{\G}$. It is easy to see that the attacker's utility in any NE equals his (unique) utility in the zero-sum game $\bar{\G}$.  

As a corollary of Lemma  \ref{lem:NEset} and Theorem \ref{thm:main},  there is a polynomial time algorithm to compute \emph{one} Nash equilibrium for security games if and only if there is a polynomial time DBR oracle. However, it is widely known that the Nash equilibrium is not unique, so there are issues of equilibrium selection. In many security settings (e.g., \cite{mavronicolas2005,mavronicolas2006}), the NE that maximizes or minimizes the defender's utility is a natural choice for analyzing the game. Though Lemma \ref{lem:NEset} shows that the attacker will derive  the same utility in any NE, the defender's utilities are generally different in different NEs (examples are given in \cite{korzhyk2011}). 

As widely known, maximizing a player's utility over Nash equilibria is NP-hard even in two-player normal-form games \cite{gilboa1989nash,conitzer2008}. It will be appealing if  the optimal NE can be efficiently computed in the security game, which is widely recognized as a successful application of game theory. Our next result shows that this is indeed the case! 
We prove that though the defender's equilibrium utility is generally a convex function of the attacker's mixed strategy, it becomes linear when restricted to the domain of the attacker \emph{equilibrium strategies}. Therefore, maximizing or minimizing the defender's Nash equilibrium utility can both be efficiently handled.

\begin{theorem}\label{thm:NE}
	There is a $\poly(n)$ time algorithm to compute
	the best and worst (for the defender) Nash equilibrium for security games  over $\E$, \emph{if and only if} there is a $\poly(n)$ time algorithm to solve the $DBR$ problem over $\E$.
	Here, by ``best/worst" we mean the NE that maximizes/minimizes the defender's utility.
\end{theorem}
\begin{proof}
	The ``only if" direction follows from Theorem \ref{thm:main}, the fact that all Nash equilibria are payoff equivalent to the minimax equilibrium in zero-sum games, and that zero-sum games are special cases of general-sum games.  For the ``if" direction, we only prove this for the case of computing the best Nash equilibrium since computing the worst Nash equilibrium is similar.
	
	The defender's equilibrium utility is a function of the attacker's strategy, which we denote as $U_{NE}(\mathbf{y})$. We first derive the function form of $U_{NE}(\mathbf{y})$, as follows:
	\begin{eqnarray*}
		U_{NE}(\mathbf{y}) &=& \max_{\mathbf{e} \in \E} U^d(\mathbf{e}, \mathbf{y}) \\
		& = & \max_{\mathbf{e} \in \E} \bigg[  \sum_{i=1}^n e_i y_i(r_i - c_i) \bigg]+ \sum_{i=1}^n y_i c_i.
	\end{eqnarray*}
	$U_{NE}(\mathbf{y})$ is a convex function since $\max_{\mathbf{e} \in \E} [  \sum_{i=1}^n e_i y_i(r_i - c_i) ]$ is convex. To compute the best equilibrium, we need to maximize this convex function over feasible $\mathbf{y}$ which is generally difficult to handle. Interestingly, we show that the term $\max_{\mathbf{e} \in \E}  \big[ \sum_{i=1}^n e_i y_i(r_i - c_i) \big] $ becomes linear in $\mathbf{y}$ when $\mathbf{y}$ is restricted to be an attacker equilibrium strategy. Let $(\mathbf{x},\mathbf{y})$ be an NE of $\G$.  Recall from Lemma \ref{lem:NEset}  that $(\mathbf{x},f(\mathbf{y}))$ is the minimax equilibrium of the zero-sum game $\bar{\G}(-\mathbf{\zeta}, -\mathbf{\rho},\mathbf{\rho},\mathbf{\zeta},\E)$. Therefore, 
	\begin{eqnarray} \nonumber
	Val(\bar{\G}) &=& U^d(\mathbf{x},f(\mathbf{y})) \\ \nonumber
	&=& \sum_{i=1}^n f_i(\mathbf{y}) x_i \big(  \rho_i - \zeta_i \big) - \sum_{i=1}^n f_i(\mathbf{y}) \rho_i \\ \label{eq:valG}
	&=& \max_{\mathbf{e} \in \E} \bigg[ \sum_{i=1}^n f_i(\mathbf{y}) e_i \big(  \rho_i - \zeta_i \big) \bigg] - \sum_{i=1}^n f_i(\mathbf{y}) \rho_i 
	\end{eqnarray}
	is the game value of $\bar{\G}$. Recalling the transform $f_i(\mathbf{y}) = \frac{1}{\lambda} \frac{r_i - c_i}{\rho_i - \zeta_i} y_i$ and utilizing Equation \eqref{eq:valG}, we have 
	\begin{align*}
	\max_{\mathbf{e} \in \E} \bigg[ \sum_{i=1}^n e_i y_i(r_i - c_i) \bigg]& =  \max_{\mathbf{e} \in \E}  \bigg[ \sum_{i=1}^n e_i \lambda f_i(\mathbf{y})(\rho_i - \zeta_i) \bigg] & \mbox{By definition of $f_i(\mathbf{y})$ }\\
	& =  \lambda \cdot \max_{\mathbf{e} \in \E} \bigg[ \sum_{i=1}^n e_i f_i(\mathbf{y})(\rho_i - \zeta_i) \bigg] & \mbox{}\\
	& =   \lambda \cdot \bigg[ Val(\bar{\G}) + \sum_{i=1}^n f_i(\mathbf{y}) \rho_i  \bigg] & \mbox{By Equation \eqref{eq:valG}}\\
	& = \lambda \cdot Val(\bar{\G}) + \sum_{i=1}^n  \frac{r_i - c_i}{\rho_i - \zeta_i} y_i  \rho_i  & \mbox{By definition of $f_i(\mathbf{y})$}
	\end{align*}
	Crucially, $Val(\bar{G})$ is a constant and does not depend on $\mathbf{y}$.  Since $\lambda=\sum_{i=1}^n  \frac{r_i - c_i}{\rho_i - \zeta_i} y_i$ is linear in $\mathbf{y}$,  so is $\max_{\mathbf{e} \in \E} \big[ \sum_{i=1}^n e_i y_i(r_i - c_i) \big]$. Therefor,  $U_{NE}(\mathbf{y})$ is linear in $\mathbf{y}$. We note again that this is true only when $\mathbf{y}$ is restricted to be an attacker NE strategy since the above derivation is not valid otherwise. Fortunately, we only need to optimize over Nash equilibria. Notice that $Val(\bar{\G})$ can be computed in $\poly(n)$ time if the DBR admits a $\poly(n)$ time algorithm (by Theorem \ref{thm:main}). 
	As a result, $U_{NE}(\mathbf{y})$ is a linear function of $\mathbf{y}$ which can be evaluated efficiently. 
	Therefore, computing the best Nash equilibrium reduces to a linear optimization over the set of all the attacker's Nash equilibrium strategies, denoted as $Y_{NE}$.  To complete the proof, we now show that $Y_{NE}$ is a polytope, and admits a $\poly(n)$ time separation oracle. Therefore, by Theorem \ref{thm:sep}, we conclude that the best (for the defender) Nash equilibrium can be computed in $\poly(n)$ time.
	
	\begin{lemma}\label{lem:Ypolytope}
		The set of the attacker's Nash equilibrium strategies $Y_{NE}$ is a polytope, and admits a $\poly(n)$ time separation oracle if the DBR problem over $\E$ can be solved in $\poly(n)$ time.
	\end{lemma} 
	\noindent \emph{Proof of Lemma \ref{lem:Ypolytope}.} First, compute an arbitrary Nash equilibrium $(\tilde{ \mathbf{x} },\tilde{ \mathbf{y}})$ of $\G$. By Lemma \ref{lem:NEset} and Theorem \ref{thm:main}, this can be done in $\poly(n)$ time. We claim that the set of attacker's Nash equilibrium strategies $Y_{NE}$ is characterized precisely by the following three sets of linear constraints on $\mathbf{y} \in \RR^n$:
	\begin{eqnarray} \label{eq:attBest}
	& & \sum_{i=1}^n y_i \big(    \rho_i \cdot [1-\tilde{  x_i }] + \zeta_i \cdot \tilde{  x_i } \big) \geq    \rho_k \cdot [1-\tilde{  x_k }] + \zeta_k \cdot \tilde{  x_k }, \qquad \qquad \quad \forall k \in [n]. \\ \label{eq:defBest}
	& & \sum_{i=1}^n y_i \big(  r_i \cdot \tilde{ x_i }+ c_i \cdot [1 - \tilde{ x_i} ] \big) \geq \sum_{i=1}^n y_i \big(  r_i \cdot e_i + c_i \cdot [1 - e_i ] \big) , \, \qquad \forall \mathbf{e} \in \E. \\ \label{eq:simplex}
	& & \sum_{i=1}^n y_i = 1 \qquad and \qquad y_i \geq 0, \, \, \, \forall i \in[n].
	\end{eqnarray} 
	Inequality \eqref{eq:attBest} restricts $\mathbf{y}$ to be an attacker best response to the defender equilibrium strategy $\tilde{\mathbf{x}}$; Inequality \eqref{eq:defBest} means that $\tilde{\mathbf{x}}$ should be a defender best response to the attacker mixed strategy $\mathbf{y}$.  
	
	We first show that these constraints are necessary. By Lemma \ref{lem:NEset}, Nash equilibria in security games are interchangeable. Therefore if $\mathbf{y}$ is an attacker equilibrium  strategy, $(\tilde{\mathbf{x}}, \mathbf{y})$ is a Nash equilibrium. Thus $\mathbf{y}$ and $\tilde{\mathbf{x}}$ are best responses to each other, as described by Inequality \eqref{eq:attBest} and \eqref{eq:defBest}.  Since $\mathbf{y}$ is a mixed strategy, thus Constraint \eqref{eq:simplex} holds.
	
	To show that these constraints are also sufficient, let $\mathbf{y}$ be any vector that satisfies these constraints. Constraint \eqref{eq:simplex} restricts $\mathbf{y}$ to be a valid mixed strategy; Inequality \eqref{eq:attBest} and \eqref{eq:defBest} induce that $(\tilde{\mathbf{x}},\mathbf{y})$ is a Nash equilibrium. To sum up, Constraints $\eqref{eq:attBest} \sim \eqref{eq:simplex}$ precisely characterize the set $Y_{NE}$.
	
	We now show that $Y_{NE}$ admits a $\poly(n)$ time separation oracle. Notice that Constraints \eqref{eq:attBest} and \eqref{eq:simplex} can be explicitly checked in $\poly(n)$ time. Constraint \eqref{eq:defBest} can be simplified as follows:
	\begin{eqnarray*}
		& & 	\sum_{i=1}^n y_i \big(  r_i \cdot \tilde{ x_i }+ c_i \cdot [1 - \tilde{ x_i} ] \big) \geq \sum_{i=1}^n y_i \big(  r_i \cdot e_i + c_i \cdot [1 - e_i ] \big) , \, \qquad \forall \mathbf{e} \in \E. \\
		& \Leftrightarrow & \sum_{i=1}^n y_i \big(  r_i -c_i \big) \cdot \tilde{ x_i }\geq \sum_{i=1}^n y_i \big(  r_i -c_i \big) \cdot e_i  , \, \qquad \forall \mathbf{e} \in \E. \\
	\end{eqnarray*} 
	The inequality can be checked by feeding $w_i = y_i \big(  r_i -c_i \big) $ into the DBR problem over $\E$ and examining whether the optimal objective value is at most $\sum_{i=1}^n y_i \big(  r_i -c_i \big) \cdot \tilde{ x_i }$, and if not, returning the optimal $\mathbf{e}$ from the DBR problem as a certificate of the violated constraint. Therefore, a $\poly(n)$ time DBR oracle yields a $\poly(n)$ time separation oracle for $Y_{NE}$. This completes our proof of Lemma \ref{lem:NEset}, as well as the proof of the theorem. 
\end{proof}
The following is a simple corollary of Theorem \ref{thm:NE}. It follows from the fact proved in Theorem \ref{thm:NE} that the defender's equilibrium utility is a linear function over the set $Y_{NE}$ of attacker equilibrium strategies and $Y_{NE}$ is a polytope. Therefore, we can convert the computation of an NE with particular defender equilibrium utility to feasibility checking of a linear system. 
\begin{corollary}
	Let $\G$ be any security game, $U_{\max}$ [$U_{\min}$] be the best [worst] defender utility among all Nash equilibria. Then for any $U \in [U_{\min},U_{\max}]$, there exists an NE of $\G$ with defender utility $U$. Moreover, such an NE can be computed in polynomial time if the DBR problem over $\E$ can be solved in polynomial time.    
\end{corollary}

\section{Consequences of the Equivalence Theorems}\label{sec:cons}
In this section we discuss some implications of these equivalence theorems. 
The following corollary of Theorem \ref{thm:main}, \ref{thm:SSE} and \ref{thm:NE} shows that the complexity of a security game is determined by the set system $\E$. 
\begin{corollary}\label{cor:equiv}
	For any set system $\E$, the following problems reduce to \emph{each other} in polynomial time: 
	
	(1) Combinatorial optimization over $\E$ for non-negative linear objectives; 
	
	(2) Solving zero-sum security games over $\E$; 
	
	(3) Computing the strong Stackelberg equilibrium for  security games over $\E$; 
	
	(4) Computing the best or worst (for the defender) Nash equilibrium for security games  over $\E$.
\end{corollary}

Corollary \ref{cor:equiv} provides a more convenient way to understand the computational complexity of security games, since the complexity of the combinatorial optimization problem over $\E$ is much easier to study and analyze. In fact, Corollary \ref{cor:equiv} simultaneously implies the computational complexity for solving various types of security games. For example, when $\E$ is any matroid set system or when the optimization over $\E$ has a min-cost max flow formulation, the equilibrium of the security game can be computed efficiently. On the other hand, when the optimization over $\E$ is  a packing problem, a Knapsack problem or a coverage problem (even vertex coverage), the equilibrium computation for these security games is NP-hard in general.

Finally, using the combinatorial characterization of security games, we can  easily recover and strengthen some known complexity results in the literature of security games, as well as resolve some open problems from previous work. For example, Xu et al. \cite{Xu14a} considered the computation of the minimax equilibrium in zero-sum spatio-temporal security games. They proved that the DBR problem there is NP-hard, but the computation of the minimax equilibrium is left open. Theorem \ref{thm:main} resolves this open question. Brown et al. \cite{Brown16a} studied the airport passenger screening game (see Section \ref{sec:SG:comb}, the example of packing problems), and proved the NP-hardness to solve the game. This also follows from Corollary \ref{cor:equiv} and the fact that the independent set problem is NP-hard.  Gan et al. \cite{gan2015security} considered security games on graphs where targets are vertices. The defender chooses a subset of vertices to patrol, by which the patrolled vertices as well as their adjacent vertices are covered. The DBR problem is to, given a non-negative weight for each vertex, compute a subset of vertices  so that the total weight of the covered vertices is maximized. This is the optimization variant of the dominating set problem, a well-known NP-hard problem. Therefore, not only computing the SSE is NP-hard (as shown by Gan et al.), our results strengthen the hardness to computing the minimax equilibrium. 
Korzhyk et al. \cite{korzhyk2010complexity} considered security games in which the DBR problem is the coverage problem as discussed in Section \ref{sec:SG:comb}. They showed polynomial solvability of security games when each (homogeneous) resource can protect a subset of at most $2$ targets (e.g., a pair of round-trip flights). This also follows from Theorem \ref{thm:main} and the fact that weighted 2-cover is polynomial time solvable. The NP-hardness for the case with sets of at most $3$ targets follows from Theorem \ref{thm:main} and the fact that 3-cover is NP-hard. Finally, Letchford and Conitzer \cite{letchford2013} proved complexity results under various conditions for security games played on graphs, where vertices are targets and each security resource can patrol a set of vertices on a path of the graph. Some of their results can be recovered using our results as well. For example, the paper considered two conditions, under any of which the minimax equilibrium of the game can be computed in polynomial time. It is easy to check that the DBR problems under these two conditions can be solved respectively by a greedy algorithm (the case of Theorem 1 in \cite{letchford2013}) and a dynamic program (the case of Theorem 2 in \cite{letchford2013}), thus the polynomial solvability also follows from our results.   
We omit further details here.

\section{Conclusions and Discussions}\label{sec:conclude}
In this paper, we systematically studied the computational complexity of equilibrium computation in security games. Our main result is the polynomial time equivalence between computing the three mostly adopted equilibrium concepts in security games, namely, the minimax equilibrium, strong Stackelberg equilibrium, best/worst Nash equilibrium, and computing the defender's best response. We believe that our results form a theoretical basis for further  algorithm design and complexity analysis in security games. 

Future research can take a number of directions. First, given that exactly solving the DBR problem is NP-hard in many cases, it is interesting to examine the approximate version of all our equivalence theorems. That is, how an approximate defender best response oracle relates to the approximate computation of an equilibrium. We note that using the no-regret learning framework, one can convert an FPTAS for the DBR problem to an algorithm for computing an  $\epsilon$-minimax equilibrium (see \cite{Immorlica2011}), but the reverse direction and other generalizations are open. Our results on computing the best/worst Nash equilibrium are surprising. Since the security game is a special class of bilinear games, we wonder whether similar results hold in other special (and interpretable) class of bilinear games. Finally, there are several ways to generalize our model. For example, the players' utility functions may  not be linear, but can still by compactly represented and computed.  One particular example is the network interdiction game played on a graph \cite{washburn1995two,tsai2010}, in which the defender chooses edges to defend and the attacker chooses a path to attack. The task of the defender is to interdict the attacker at a certain edge. This is not captured by our (bilinear) framework. Another generalization is to allow the attacker to attack multiple targets \cite{korzhyk2011security}. We wonder how the computational complexity of the proposed four problems relates to each other in these generalized settings.

\section*{ACKNOWLEDGMENTS}
We would like to thank Shaddin Dughmi, Milind Tambe and Vincent Conitzer for helpful discussions. We also thank the anonymous EC reviewers for helpful feedback and suggestions.

	\bibliography{refers}
	\bibliographystyle{plain}

\end{document}